\documentclass[letterpaper, 10 pt,journal]{ieeetran}
\IEEEpubidadjcol 
\IEEEoverridecommandlockouts  

\usepackage[utf8]{inputenc}

\usepackage{amsmath}
\usepackage{amssymb, amsthm}
\usepackage[T1]{fontenc}
\usepackage{graphicx, xcolor, algpseudocode, algorithm, comment, graphicx, epsfig}

\DeclareMathOperator*{\argmin}{arg\,min}
\DeclareMathSymbol{\shortminus}{\mathbin}{AMSa}{"39}

\def\diag{\text{diag}}
\def\bquad{\!\!\!\!}
\def\LMS{least-mean square }

\usepackage{cite}

\usepackage[hidelinks]{hyperref}
\newtheorem{theorem}{Theorem}
\newtheorem{lemma}[theorem]{Lemma}
\newtheorem{remark}{Remark}
\newtheorem{corollary}[theorem]{Corollary}
\newtheorem{proposition}[theorem]{Proposition}
\newtheorem{assumption}{Assumption}

\begin{document}
\IEEEoverridecommandlockouts

\title{Adaptive Economic Model Predictive Control for linear systems with performance guarantees}
\author{Maximilian Degner$^{1,2}$, Raffaele Soloperto$^1$, Melanie N. Zeilinger$^2$, John Lygeros$^1$, Johannes K\"ohler$^2$
\thanks{$^1$Automatic Control Laboratory, ETH Zürich}
\thanks{$^2$Institute for Dynamic Systems and Control, ETH Zürich}
\thanks{Johannes K\"ohler was supported by the Swiss National Science Foundation under NCCR Automation (grant agreement 51NF40 180545).}
\thanks{Raffaele Soloperto is supported by
the European Research Council under the H2020 Advanced Grant no. 787845 (OCAL).}
}
\maketitle
\begin{abstract}
We present a model predictive control (MPC) formulation to directly optimize economic criteria for linear constrained systems subject to disturbances and uncertain model parameters. 
The proposed formulation combines a certainty equivalent economic MPC with a simple least-squares parameter adaptation. 
For the resulting adaptive economic MPC scheme, we derive strong asymptotic and transient performance guarantees.
We provide a numerical example involving building temperature control and demonstrate performance benefits of online parameter adaptation. 
\end{abstract}
\IEEEpubid{\begin{minipage}{\textwidth}\ \\[24pt] \\ \\
         \copyright 2024 IEEE.  Personal use of this material is  permitted.  Permission from IEEE must be obtained for all other uses, in  any current or future media, including reprinting/republishing this material for advertising or promotional purposes, creating new  collective works, for resale or redistribution to servers or lists, or  reuse of any copyrighted component of this work in other works.
     \end{minipage}}
\section{Introduction}
Model predictive control (MPC) is an optimization-based control strategy that is applicable to general MIMO systems, explicitly accounts for constraints, and optimizes general performance criteria~\cite{rawlingsModelPredictiveControl2020}. 
However, deployment strategies of MPC typically rely on a sequential design and hierarchical architecture that might limit performance: 
The prediction model is obtained from physical first principles or system identification tools. During online operation, a real-time optimization layer computes optimal setpoints that are tracked by an MPC scheme~\cite{rawlingsFundamentalsEconomicModel2012}.
Changing environments or wear can deteriorate the accuracy of the offline obtained model, and thus jeopardize the performance of the MPC.
Furthermore, simply tracking a fixed setpoint limits the overall efficiency in case of persistent disturbances or frequently changing operating conditions.
This paper addresses both issues by directly optimizing economic performance and using online model adaptation.  

\subsection*{Related work} 
\emph{Economic} MPC formulations~\cite{rawlingsFundamentalsEconomicModel2012,faulwasser2018economic,kohler2020periodic,faulwasser2019toward} directly minimize the economic objective over the prediction horizon to improve performance. 
Exemplary applications for this methodology are minimizing energy consumption in heating, ventilation, and air conditioning (HVAC)~\cite{taheriModelPredictiveControl2022}, increasing production in chemical plants~\cite{rawlingsFundamentalsEconomicModel2012} or flexible manufacturing system~\cite{risbeck2019unification}. 
One of the key theoretical results in economic MPC is that, on average, the closed-loop performance is no worse than the performance at the optimal steady-state~\cite{angeliAveragePerformanceStability2012a,amritEconomicOptimizationUsing2011}, i.e., at worst we obtain an asymptotic performance comparable to a tracking MPC. 
Stronger transient performance bounds relative to the infinite-horizon optimal solution are obtained in~\cite{grune2014asymptotic,grune2015non}, see also~\cite{bayer2018optimal,schwenkel2020robust,kloppelt2021transient} for robust performance bounds under bounded disturbances. 

\emph{Adaptive} MPC uses online measurements to adapt the model parameters, thus relying less on the initial parameter identification and ensuring reliable operation under online changes in system behavior. 
While the benefits of adaptive MPC methods have been demonstrated empirically~\cite{yoon1994adaptive}, the development of adaptive MPC theory has received less attention~\cite{mayne2014model}.
\emph{Robust adaptive} MPC schemes~\cite{lorenzenRobustMPCRecursive2019a,kohlerLinearRobustAdaptive2020,lu2021robust} combine a robust MPC design with set-membership  updates to reduce conservatism in constraint handling, see also~\cite{gonccalves2016robust,kohlerRobustAdaptiveModel2021a,sasfi2023robust} for nonlinear systems.
In~\cite{lorenzenRobustMPCRecursive2019a}, finite-gain $\mathcal{L}_2$ stability is shown by updating parameters in the cost with a \LMS  filter. 
In~\cite{kohlerRobustAdaptiveModel2021a}, it is shown that the same stability result applies to linearly parameterized nonlinear systems under suitable regularity conditions. 
Convergence results for more general non-parametric model updates are shown in~\cite{solopertoGuaranteedClosedLoopLearning2023,zuliani2023convergence}.

\emph{Adaptive economic MPC:}  
The integration of online measurements to  ensure optimality of the steady-state is studied in the framework of modifier adaptation~\cite{marchetti2009modifier}. 
MPC approaches that integrate adaptation to improve the steady-state optimization are provided in~\cite{shaltoutAdaptiveEconomicModel2018,faulwasser2019toward,vaccari2021offset,oliveira-silvaEconomicMPCModifier2023}.  
A nonlinear robust adaptive economic MPC scheme is proposed in~\cite[Chap.~3]{wu2023learning}, however, no performance guarantees are derived. 
Overall, theoretical performance guarantees for economic MPC with online model adaptation require further study. 

\subsection*{Contribution}
In this paper, 
we consider linear constrained, open-loop stable systems subject to disturbances and uncertain model parameters with the goal to minimize a convex linear-quadratic economic cost. 
Our main assumptions are open-loop stability of the linear system and relaxation of state constraints with soft penalties. 
For this problem, we propose an adaptive economic MPC (AE-MPC) scheme which:
\begin{itemize}
    \item uses a certainty-equivalent economic MPC scheme with a linear-quadratic terminal cost in combination with the LMS parameter adaptation from~\cite{lorenzenRobustMPCRecursive2019a};
    \item for finite-energy disturbances,  recovers the asymptotic performance guarantees from economic MPC with perfect model knowledge~\cite{angeliAveragePerformanceStability2012a}, despite the possibly large initial error in the model parameters (Theorem~\ref{thm_asymp-perf});
    \item for the case of bounded disturbances, ensures a more general transient performance bound (Theorem~\ref{thm_transient-perf}) that depends linearly on the magnitude of the disturbances.
\end{itemize}

We demonstrate the performance benefits of the proposed AE-MPC scheme with a numerical example involving building temperature control in comparison to an economic MPC without parameter adaptation.

\subsection*{Outline and Notation}
We first define the problem setup (Sec.~\ref{sec_problem-setup}) and present the proposed AE-MPC scheme (Sec.~\ref{sec_method}). 
Then we provide a theoretical performance analysis and discussion (Sec.~\ref{sec_theoretical-results}), 
and demonstrate these results with a simulation of a simple building temperature control problem (Sec.~\ref{sec_numerical-example}). 

We denote the $1$-norm and $2$-norm of a vector $x\in\mathbb{R}^n$ by $\|x\|_1$ and $\|x\|$, respectively. 
The weighted norm with $Q\in\mathbb{R}^{n\times n}$ is denoted by $\|x\|_Q^2=x^\top Qx$.  
The spectral norm of a matrix $A\in\mathbb{R}^{n\times m}$ is given by $\|A\|=\sqrt{\lambda_{\text{max}}( A^\top A)}$, where $\lambda_{\max}$ denoted the maximal eigenvalue.
A positive semidefinite matrix $R$ is denoted by $R\succeq 0$ and a positive definite matrix $Q$ is denoted by $Q\succ0$.
The identity matrix of size $n\times n$ is signified by $I_n$.
At time $k$, we denote a prediction of quantity $x$ for time $k+j$ as $x_{j|k}$. 
The $i$-th element of a vector $x\in\mathbb{R}^n$ is represented by $[x]_i$.
A function $\alpha$ is of class $\mathcal{K}_\infty$, if $\alpha: \mathbb{R}_{\geq0}\rightarrow\mathbb{R}_{\geq0}, \alpha(0)=0, \lim_{s\rightarrow \infty}\alpha(s)=\infty$, and $\alpha$ strictly increasing and continuous. 
The non-negative integers are represented by $\mathbb{N}$ and $\mathbb{N}_{[a,b]} = {\{n\in\mathbb{N}| a\leq n\leq b\}}$.

\section{Problem setup} \label{sec_problem-setup}
We consider discrete-time linear systems that are subject to additive disturbances and parametric uncertainty. The dynamics are given by
\begin{align} \label{DT-system-equation}
    x_{k+1} = A(\theta^\ast) x_k + B(\theta^\ast)u_k + w_k,
\end{align}
where $x_k\in \mathbb{R}^n,~ u_k\in\mathbb{R}^m, ~ w_k\in\mathbb{R}^n$ are the states, control inputs, and disturbances at time $k\in \mathbb{N}$, respectively. The parameter $\theta\in\mathbb{R}^d$ enters~\eqref{DT-system-equation} in an affine way, i.e., 
    \begin{equation*}
        A(\theta) = A_0 + \sum_{i=1}^d A_i\cdot [\theta]_i,\ B(\theta) = B_0 + \sum_{i=1}^d B_i\cdot [\theta]_i,
    \end{equation*}
where $A_i\in\mathbb{R}^{n\times n}, B_i\in\mathbb{R}^{n\times m},\ i\in\mathbb{N}_{[0,d]}$. The true parameter $\theta^\ast$ is unknown but constant over time and lies in a known compact, polytopic set $\Theta\subset \mathbb{R}^d$.
We assume that we have an initial parameter estimate $\hat{\theta}_0\in\Theta.$

We study the case where we want to satisfy the following state and input constraints
 \begin{align} \label{problem_setup_compact-constraints}
     (x_k, u_k)\in \mathbb{X} \times \mathbb{U}, \quad \forall k\in\mathbb{N},
 \end{align}
where $\mathbb{X}=\{x\in\mathbb{R}^n|Hx\leq h\}$, $h\in \mathbb{R}^c_{\geq 0}$, $\mathbb{U}$ is a compact polytope, and $(0,0)\in \mathbb{X}\times\mathbb{U}$.
The goal is to minimize the following convex economic cost
\begin{align}
    \ell_{\text{eco}}(x,u) &= \|x\|^2_Q + \|u\|^2_R + q^\top x + r^\top u, \label{problem_setup_eco-cost}
\end{align}
with $Q,~R \succeq 0$, which is in general not positive definite with respect to the origin.
To simplify the design, we soften the state constraints with a slack variable $s = \max\{Hx-h,0\}\in\mathbb{R}^c_{\geq 0}$, which is penalized by using the stage cost
\begin{align}
\ell(x,u,s)&=\ell_{\text{eco}}(x,u)+ \|s\|_\Lambda^2,
\label{problem_setup_eco-cost_2}
\end{align}
with a slack penalty $\Lambda\succ 0$. For simplicity of exposition, we consider a diagonal matrix $\Lambda$. 
\begin{assumption}[Stable system] \label{assump_system-matrices}
   The system is open-loop exponentially stable for all $\theta\in\Theta$ with a common Lyapunov function, i.e., $\exists P\succ 0$: 
   \begin{equation*}
        A(\theta)^\top P A(\theta)+I\preceq P \ \forall \theta\in\Theta.
    \end{equation*}
\end{assumption}
\begin{remark}[Open-loop stability and soft-constraints]
By restricting the problem setup to open-loop stable systems (Assumption~\ref{assump_system-matrices}) and soft-state constraints, the proposed method can be applied even if the initial parameter error and the disturbances are large. 
Notably, if sufficiently small bounds on parametric error and disturbances are known, both simplifications can be relaxed, see the discussion in Section~\ref{sec_discuss}.
\end{remark}

We assume that some possibly conservative bound on the disturbances is known, i.e., $w_k\in\mathbb{W}$, $\forall k\in\mathbb{N}$ with $\mathbb{W}$ compact. 
In combination with $\mathbb{U}$ compact, $A(\theta^*)$ Schur stable (Assumption~\ref{assump_system-matrices}) and some finite $x_0$, this implies that the closed-loop states $x_k$ are uniformly bounded. Hence, we assume that some conservative upper-bound on the state is available, i.e., $x_k\in\mathbb{Z}$, $\forall k\in\mathbb{N}$ with a known compact set $\mathbb{Z}$.

\section{Proposed Approach} \label{sec_method}
In this section, we describe the offline design and the online algorithm of the proposed adaptive economic MPC (AE-MPC) formulation. First, we provide the parameter adaptation scheme and its properties (Sec.~\ref{ssec_parameter-adapt}), then we construct a terminal cost (Sec.~\ref{ssec-terminal-cost}), and finally we state the resulting AE-MPC scheme (Sec.~\ref{ssec_AEMPC-scheme}).

\subsection{Parameter adaptation} \label{ssec_parameter-adapt}
Based on the initial estimate $\hat{\theta}_0\in\Theta$ of the true parameter $\theta^\ast$, we use the projected \LMS (LMS) filter from~\cite{lorenzenRobustMPCRecursive2019a} to adapt the parameter estimate $\hat{\theta}_k$ at every time instant $k\in\mathbb{N}$. 
The one-step ahead prediction at time $k$, is given by $x_{1|k}=A_0 x_k + B_0 u_k + D_k\hat{\theta}_k$, where $D_k = D(x_k,u_k) \in \mathbb{R}^{n \times d}$ and
\begin{align*}
    D(x_k,u_k) =\! [A_1x_k+B_1u_k \dots A_d x_k+B_d u_k].
\end{align*}

The update equations of the LMS are 
\begin{subequations}\label{LMS_update-eqs}\begin{align} 
    \begin{split}
        \tilde{\theta}_{k+1} &= \hat{\theta}_k + \mu D_k^\top (x_{k+1}-x_{1|k})\\
    &= \hat{\theta}_k + \mu D_k^\top (D_k\cdot(\theta^\ast - \hat{\theta}_k) + w_k),\label{LMS_update_a}
    \end{split}\\
    \hat{\theta}_{k+1} &= \argmin_{\theta \in \Theta} \|\theta-\tilde{\theta}_{k+1}\|, \label{LMS_update_b}
\end{align} \end{subequations}
where $\mu>0$ is the update gain which is chosen such that 
\begin{align}\label{LMS_gain-definition}
    \frac{1}{\mu} \geq \sup_{(x,u)\in \mathbb{Z}\times \mathbb{U}} \|D(x,u)\|^2.
\end{align} The following proposition is an adaption of~\cite[Lemma~5]{lorenzenRobustMPCRecursive2019a} and shows that the cumulated prediction error is bounded.
\begin{proposition}[LMS bounds]\label{prop_LMS-bound}
    Suppose that $x_k\in \mathbb{Z}$ and $u_k\in\mathbb{U}$ for all $k\in\mathbb{N}$. Then, for all $T\in\mathbb{N}$, it holds that
    \begin{align}\label{prop_LMS-bound_equation}
        \sum_{k=0}^T \|\tilde{x}_{1|k}\|^2\leq \frac{1}{\mu}\|\hat{\theta}_0-\theta^\ast\|^2 + \sum_{k=0}^T \|w_k\|^2,
    \end{align}
    with the one-step parametric prediction error
    \begin{align*}
        \tilde{x}_{1|k} := D_k\cdot(\theta^\ast-\hat{\theta}_k).
    \end{align*}
    Moreover, the difference between two successive parameter estimates satisfies
    \begin{align} \label{lemma_theta_diff_equation}
    \|D(x,u)\|\cdot \|\hat{\theta}_{k+1}-\hat{\theta}_k\| \leq \|\tilde{x}_{1|k}+w_k\|,~\forall (x,u) \in \mathbb{Z}\times \mathbb{U}.
    \end{align}
\end{proposition}
\begin{proof}
    The proof of~\eqref{prop_LMS-bound_equation} can be found in \cite[Lemma~5]{lorenzenRobustMPCRecursive2019a}. Inequality \eqref{lemma_theta_diff_equation} stems from the fact that the projection operator is non-expansive, and is derived using the update equation and the Cauchy-Schwarz inequality: 
    \begin{align} \begin{split}
    \|\hat{\theta}_{k+1} - \hat{\theta}_k\|  \leq 
    \|\tilde{\theta}_{k+1}-\hat{\theta}_k\| \stackrel{\eqref{LMS_update_a}}{\leq} \mu \|D_k\| \cdot \|\tilde{x}_{1|k}+w_k\|. \label{prop_LMS_delta-theta-relation}
    \end{split}\end{align}
    Multiplying both sides with $\|D(x,u)\|$ and using our choice of $\mu$ (cf.~\eqref{LMS_gain-definition}) gives the result.
    \end{proof}
\subsection{Terminal cost design} \label{ssec-terminal-cost}
MPC schemes require a terminal cost for closed-loop performance guarantees~\cite{rawlingsModelPredictiveControl2020, amritEconomicOptimizationUsing2011}. In this section, we extend the standard design procedure in~\cite{amritEconomicOptimizationUsing2011} to incorporate online parameter adaptation and show how to design a simple terminal cost $\ell_\mathrm{f}(x,\theta)$. 
We center the design of terminal cost around the steady-state $(x,u)=0$ for simplicity.
Hence, all theoretical guarantees will be relative to the performance at this steady state. 

\begin{proposition}[Terminal cost] \label{prop_term-cost}
Consider $P_\mathrm{f}\in\mathbb{R}^{n\times n}$, $P_\mathrm{f}\succ 0$ and the vector-valued function $p: \mathbb{R}^d\rightarrow\mathbb{R}^n$ satisfying:
    \begin{subequations}\begin{align}
        & A(\theta)^\top P_\mathrm{f}A(\theta) - P_\mathrm{f} + \bar{Q} \preceq 0, \quad\forall \theta\in\Theta,\label{cost-function_design-a}\\
        & p(\theta):=\left[(I_n-A(\theta))^{-1}\right]^{\top} q, \label{cost-function_design-b}
    \end{align}\end{subequations}
    with $\bar{Q} = Q + H^\top \Lambda H$. Then, for all $x\in\mathbb{R}^n,\theta\in\Theta$, the terminal cost $\ell_\mathrm{f}(x,\theta) = \|x\|^2_{P_\mathrm{f}} + p(\theta)^\top x$ satisfies
    \begin{align}\label{prop_term-cost_equation}
        \ell_{\mathrm{f}}(A(\theta)x, \theta)-\ell_{\mathrm{f}}(x,\theta) \leq -\ell(x,0,s),
    \end{align}
    with $s=\max\{Hx-h,0\}$.
\end{proposition}

\begin{proof}
    First, note that $\|\max\{Hx-h,0\}\|_\Lambda^2 \leq \|Hx\|^2_\Lambda$ for all $x\in \mathbb{R}^n$ with the diagonal matrix $\Lambda$. 
    Thus, 
    \begin{align}
        \ell(x,u,s) \leq \|x\|^2_{\bar{Q}} + \|u\|^2_R + q^\top x + r^\top u, \label{proof_term-cost_bound}
    \end{align}
    with $\bar{Q} = Q + H^\top \Lambda H$. To show that~\eqref{prop_term-cost_equation} holds, we use the shorthand $A_\theta=A(\theta)$ and expand the left side:
    \begin{align*}
         &\ell_{\mathrm{f}}(A_\theta x,\theta)-\ell_{\mathrm{f}}(x, \theta) \\
         &\qquad= \|A_\theta x\|^2_{P_\mathrm{f}} + p(\theta)^\top A_\theta x-\|x\|^2_{P_\mathrm{f}} - p(\theta)^\top x\\
         &\qquad\! \stackrel{\eqref{cost-function_design-b}}{=}-q^\top x+(A_\theta x)^\top P_\mathrm{f} (A_\theta x) - x^\top P_\mathrm{f} x \\
         &\qquad\!\stackrel{\eqref{cost-function_design-a}}{\leq} -q^\top x- x^\top \bar{Q}x \stackrel{\eqref{proof_term-cost_bound}}{\leq} -\ell(x,0,s). \qedhere
    \end{align*}
\end{proof}
Due to Assumption~\ref{assump_system-matrices}, we can compute a matrix $P_\mathrm{f}$ satisfying~\eqref{cost-function_design-a} using a semi-definite program with the vertices of the polytope $\Theta$, see, e.g.~\cite[Eq.~(33)]{kohlerLinearRobustAdaptive2020}. The vector $p(\theta)$ needs to be computed online for each parameter estimate~$\hat{\theta}_k$.

\subsection{Adaptive Economic MPC scheme} \label{ssec_AEMPC-scheme}
At time $k\in\mathbb{N}$, the proposed AE-MPC is given by 
\begin{subequations}\label{MPC-scheme_opt-problem}\begin{alignat}{2} 
    &\min_{u_{\cdot|k},\hat{x}_{\cdot|k},s_{\cdot|k}}&& \sum_{j=0}^{N-1}\! \ell(\hat{x}_{j|k},u_{j|k}, s_{j|k})\! +\! \ell_\mathrm{f}(\hat{x}_N, \hat{\theta}_k) \label{MPC-scheme_opt-problem-a}\\
    &\text{subject to}\quad &&\hat{x}_{j+1|k} = A(\hat{\theta}_k)x_{j|k}+B(\hat{\theta}_k)u_{j|k} \label{MPC-scheme_opt-problem-b}\\
    & \quad &&H\hat{x}_{j|k}\leq h+s_{j|k} \label{MPC-scheme_opt-problem-c}\\
    & \quad &&u_{j|k}\in\mathbb{U} \label{MPC-scheme_opt-problem-d}\\
    & \quad &&\hat{x}_{0|k} = x_k\label{MPC-scheme_opt-problem-e} \\
    &\quad && \quad\qquad \qquad \qquad \forall\, j\in\mathbb{N}_{[0,N-1]}. \notag
\end{alignat} \end{subequations}
With~\eqref{MPC-scheme_opt-problem-b}, we predict the system's trajectory with the LMS estimate and initialize these predictions with the measurement of the current state~\eqref{MPC-scheme_opt-problem-e} over the prediction horizon $N\in\mathbb{N}$. The soft state constraints are given in~\eqref{MPC-scheme_opt-problem-c}, the input constraints are imposed by~\eqref{MPC-scheme_opt-problem-d}, and we minimize the economic cost and the penalty of slack variables together with the terminal cost by using the cost function~\eqref{MPC-scheme_opt-problem-a}. The solution of the optimization problem~\eqref{MPC-scheme_opt-problem} is denoted by $u_{\cdot|k}^\ast,\hat{x}_{\cdot|k}^\ast,s_{\cdot|k}^\ast$ and the corresponding minimum of the cost function by $V_N^\ast(x_k,\hat{\theta}_k)$. The optimal control input $u_{0|k}^\ast$ is applied to the system~\eqref{DT-system-equation}. 
Note that Problem~\eqref{MPC-scheme_opt-problem} is a convex quadratic program which is feasible for all $(x_k,\hat{\theta}_k)$ since no hard state or terminal set constraints are imposed. 

A summary of the offline design and the online computations is given by Algorithm~\ref{AEMPC-algorithm}.
\begin{algorithm}
\caption{Adaptive Economic MPC}\label{AEMPC-algorithm}
\begin{algorithmic}[0]
\State Choose $\mu$ as in~\eqref{LMS_gain-definition}, compute $P_\mathrm{f}$ as in~\eqref{cost-function_design-a}.
\For{$k\in\mathbb{N}$}
    \State Measure the state $x_k$.
    \State Adapt the parameter $\hat{\theta}_k \in \Theta$ using LMS~\eqref{LMS_update-eqs}.
    \State Update linear terminal cost $p^\top \gets q^\top (I_n-A(\hat{\theta}_k))^{-1}$.
    \State Solve the optimization problem~\eqref{MPC-scheme_opt-problem}.
    \State Apply the control input $u_k \gets u^\ast_{0|k}$.
\EndFor
\end{algorithmic}
\end{algorithm}

\section{Theoretical Analysis} \label{sec_theoretical-results}
In this section, we present asymptotic and transient performance guarantees for asymptotic and transient bounds on the disturbances $w_k$. 
First, we derive an asymptotic performance bound assuming finite-energy disturbances (Sec.~\ref{sec_theory_asymptotic}). Then, we consider point-wise bounded disturbances and obtain a transient performance bound (Sec~\ref{sec_theory_trans}). Finally, we provide a discussion on the theoretical properties (Sec.~\ref{sec_discuss}).

\subsection{Asymptotic performance}
\label{sec_theory_asymptotic}
In the following, we provide an asymptotic average performance bound for finite-energy disturbances.
\begin{assumption}[Finite-energy disturbances]\label{assump_finite-energy}
    There exists a finite constant $S_\mathrm{w}$, such that
    \begin{align} \label{assump_finite-energy_equation}
        \lim_{T\rightarrow\infty}\sum_{k=0}^T \|w_k\|^2 \leq S_{\mathrm{w}}.
    \end{align}
\end{assumption}
\begin{theorem}[Asymptotic average performance] \label{thm_asymp-perf}
    Let Assumptions~\ref{assump_system-matrices} and~\ref{assump_finite-energy} hold. 
    Then, the closed-loop system resulting from Algorithm~\ref{AEMPC-algorithm} satisfies
    \begin{align} \label{thm_asymp-perf_equation}
        \limsup_{T\rightarrow\infty} \frac{\sum_{k=0}^{T-1} \ell(x_k,u_k, s_{0|k}^\ast)}{T} \leq 0.
    \end{align}
\end{theorem}

\begin{proof}
We first define a feasible candidate solution and provide a bound on the one-step decrease of the value function in~\eqref{proof_asyp_p2-result}. Then, we use the LMS properties (Prop.~\ref{prop_LMS-bound}) and the finite-energy of the disturbances (Asm.~\ref{assump_finite-energy}) to show~\eqref{thm_asymp-perf_equation}.

\textit{Candidate solution:} Recall that $\hat{x}_{\cdot|k}^\ast$, $u^\ast_{\cdot|k}$, $s^\ast_{\cdot|k}$, denote the optimal solution to Problem~\eqref{MPC-scheme_opt-problem} at time $k\in\mathbb{N}$ and let us denote $u^\ast_{N|k}=0$, $\hat{x}^\ast_{N+1|k}=A(\hat{\theta}_k)\hat{x}^\ast_{N|k}$.
We consider the candidate input sequence ${u}_{j|k+1} = u_{j+1|k}^\ast$, $j\in\mathbb{N}_{[0,N-1]}$.
The corresponding state sequence $\hat{x}_{j|k+1}$, $j\in\mathbb{N}_{[0,N]}$ is given by $\hat{x}_{0|k+1} = x_{k+1}$ and using~\eqref{MPC-scheme_opt-problem-b} with $\hat{\theta}_{k+1}$; the slack variables are chosen as $s_{j|k+1}=\max\{H\hat{x}_{j|k+1}-h,0\}$. 
We denote $\ell^\ast_{j|k} = \ell(\hat{x}_{j|k}^\ast,u^\ast_{j|k}, s_{j|k}^\ast)$, ${\ell}_{j|k+1} = \ell(\hat{x}_{j|k+1},{u}_{j|k+1}, s_{j|k+1})$.

\textit{Difference in cost:} 
    The increase of the value function $V_N^\ast$ can be upper bounded by the cost of the (feasible) candidate sequence:
    \begin{align}\label{proof_asymp_cost-diff}
         &V_N^\ast(x_{k+1},\, \hat{\theta}_{k+1}) - V_N^\ast(x_k, \hat{\theta}_k)+\ell_{0|k}^\ast \nonumber \\ 
         \leq& \sum_{j=1}^{N-1}( \ell_{j-1|k+1}-\ell^\ast_{j|k} )+ {\ell}_{N-1|k+1} \\
         & + \ell_{\mathrm{f}}(\hat{x}_{N|k+1},\hat{\theta}_{k+1}) - \ell_{\mathrm{f}}(\hat{x}_{N|k}^\ast,\hat{\theta}_{k}).\nonumber
    \end{align}      
    In the following, we first bound the increase due to disturbances and parameter update, and then bound the last terms using the properties of the terminal cost.
    Let $\Delta x_{j|k}=\hat{x}_{j-1|k+1}- \hat{x}_{j|k}^\ast$, ${\Delta \theta_k = \hat{\theta}_{k+1}-\hat{\theta}_k}$, 
    and recall that the states, inputs, and parameters all lie in some compact sets.
    Hence, the linear quadratic stage cost from~\eqref{problem_setup_eco-cost_2} has a uniform Lipschitz constant $L\geq 0$:
    \begin{align}
        &\Delta \ell_{j|k} := {\ell}_{j-1|k+1}-\ell^\ast_{j|k} \leq L \|\Delta x_{j|k}\|. \label{proof_asymp_Lipschitz-stage}
    \end{align}
    Similarly, $\ell_\mathrm{f}$ is Lipschitz-continuous\footnote{
    Note that Condition~\eqref{cost-function_design-a} ensures stability of $A(\theta)$ and that ${\|(I-A(\theta))^{-1}\|}$ is uniformly bounded. Given that also $A(\theta)$ is affine in $\theta$, we have that $p(\theta)$ in~\eqref{cost-function_design-b} is Lipschitz continuous w.r.t. $\theta\in\Theta$. Thus,  $\ell_{\mathrm{f}}$ is Lipschitz w.r.t. $(x,\theta)$ on the compact set $(x,\theta)\in\mathbb{Z}\times\Theta$.} with some uniform constant $L_{\mathrm{f}}\geq 0$, which implies
    \begin{align} 
        &\ell_{\mathrm{f}}(\hat{x}_{N|k+1},\hat{\theta}_{k+1}) - \ell_{\mathrm{f}}(\hat{x}_{N|k}^\ast,\hat{\theta}_k)         \label{proof_asymp_term-bound}\\
        \stackrel{\eqref{prop_term-cost_equation}}{\leq} &
        \ell_{\mathrm{f}}(\hat{x}_{N|k+1},\hat{\theta}_{k+1}) -\ell_{\mathrm{f}}(\hat{x}^\ast_{N+1|k},\hat{\theta}_{k}) -\ell_{N|k}^\ast \notag\\
        \leq\,& L_\mathrm{f} (\|\Delta x_{N+1|k}\|+\|\Delta \theta_k\|) - \ell_{N|k}^\ast,\notag
    \end{align}
    where we denoted $\ell_{N|k}^\ast=\ell(\hat{x}_{N|k}^\ast, 0, s^\ast_{N|k})$ and $s^\ast_{N|k} =\max\{Hx_{N|k}^\ast-h,0\}$.
    Combining~\eqref{proof_asymp_cost-diff}--\eqref{proof_asymp_term-bound} gives
    \begin{align}
    &V_N^\ast(x_{k+1}, \hat{\theta}_{k+1}) - V_N^\ast(x_k, \hat{\theta}_k)+\ell_{0|k}^\ast\notag \\
    \notag
    & \stackrel{\eqref{proof_asymp_cost-diff}}{\leq}\! \sum_{j=1}^{N-1} \! \Delta \ell_{j|k} \! + \ell_{N-1|k+1} \! + \! \ell_{\mathrm{f}}(\hat{x}_{N|k+1},\hat{\theta}_{k+1}) \! - \ell_{\mathrm{f}}(\hat{x}_{N|k}^\ast,\hat{\theta}_k)\\
    \notag
    &\stackrel{\eqref{proof_asymp_term-bound}}{\leq} \sum_{j=1}^{N} \Delta \ell_{j|k} + L_\mathrm{f} \|\Delta x_{N+1|k}\| 
     +L_{\mathrm{f}}\|\Delta \theta_k\|\\ 
    &\stackrel{\eqref{proof_asymp_Lipschitz-stage}}{\leq } L_\mathrm{f} \|\Delta x_{N+1|k}\| +L_{\mathrm{f}}\|\Delta \theta_k\|+L \cdot \sum_{j=1}^{N} \|\Delta x_{j|k}\|. \label{proof_asyp_p2-result} 
    \end{align}

    \textit{Combination with the LMS estimator:} Next, we use the LMS properties to further bound the terms in~\eqref{proof_asyp_p2-result}. Similar to~\cite{lorenzenRobustMPCRecursive2019a}, we can bound $\Delta x_{j|k}$: 
    \begin{align} \begin{split} \label{proof_asymp_delta-x}
        \Delta x_{j|k} &= A(\hat{\theta}_{k+1})^{j\shortminus1}(w_{k}+\tilde{x}_{1|k})\\
        &\quad +\sum_{i=1}^{j-1} A(\hat{\theta}_{k+1})^{j\shortminus 1\shortminus i}D(\hat{x}_{i|k},u_{i|k})\cdot\Delta \theta_k, \end{split}\\[2pt]
    \|\Delta x_{j|k}\| \quad &\bquad \stackrel{\eqref{proof_asymp_delta-x},\eqref{lemma_theta_diff_equation}}{\leq}  \sum_{i=0}^{j-1} \|A(\hat{\theta}_{k+1})^{j-1-i}\| \cdot\|w_{k}+\tilde{x}_{1|k}\|. \label{proof_asymp_bound-delta-x}
\end{align}
   Let us denote
   \begin{align} \label{proof_asymp_C_A-def}
       C_A' := \max_{j\in\mathbb{N}_{[1,N+1]}}\max_{\theta\in\Theta} \sum_{i=0}^{j-1} \|A(\theta)^{j-1-i}\|.
   \end{align}
    Combining the bounds and summing both sides from $k=0$ to $k=T-1$ yields  
    \begin{align} 
        &V_N^\ast(x_T, \hat{\theta}_T)-V_N^\ast(x_0, \hat{\theta}_0) + \sum_{k=0}^{T-1}\ell_{0|k}^\ast \notag \\ 
        & \stackrel{\eqref{proof_asyp_p2-result}}{\leq} \sum_{k=0}^{T-1}\left[L_\mathrm{f} \|\Delta x_{N+1|k}\|+L_{\mathrm{f}}\|\Delta \theta_k\| +L \sum_{j=1}^{N} \|\Delta x_{j|k}\|\right] \notag\\
        &\stackrel{\eqref{prop_LMS_delta-theta-relation}, \eqref{LMS_gain-definition}}{\leq} \sum_{k=0}^{T-1} \left[ L_\mathrm{f} \|\Delta x_{N+1|k}\| + L_{\mathrm{f}}\sqrt{\mu}\|\tilde{x}_{1|k} + w_k\| \vphantom{\sum_j^N} \right. \nonumber \\
         &\qquad \left.+ L \sum_{j=1}^{N}\|\Delta x_{j|k}\|\right] \notag\\
        & \stackrel{\eqref{proof_asymp_bound-delta-x},\eqref{proof_asymp_C_A-def}}{\leq} C_A\sum_{k=0}^{T-1} \left[  \|w_k\|+\|\tilde{x}_{1|k}\|\right],  \label{proof_asymp_intermediate-result}
    \end{align} 
    where the last step used the triangular inequality and $C_A:=\left[L_\mathrm{f}+(N-1)L\right]C_A'+L_\mathrm{f}\sqrt{\mu}$. 
    Note that 
    \begin{equation*}
        \lim_{T\rightarrow\infty} \frac{-V_N^\ast(x_T,\hat{\theta}_T)+V_N^\ast(x_0,\hat{\theta}_0)}{T}=0,
    \end{equation*}
    since $V_N^\ast(x_k,\hat{\theta}_k)$ admits a uniform bound using $x_k\in \mathbb{Z},~\hat{\theta}_k\in\Theta$, and compactness of $\mathbb{Z}$, $\Theta$. Thanks to Assumption~\ref{assump_finite-energy}, Lemma~\ref{lemma_asymp-lin-bound} from the Appendix ensures
    \begin{align} \label{lemma_asymp-lin-bound_equation}
         \lim_{T \rightarrow \infty} \frac{\sum_{k=0}^{T-1} \|\tilde{x}_{1|k}\| + \|w_k\|}{T} =0.
    \end{align} 
    Thus, taking the limit $T\rightarrow\infty$ gives the final result
    \begin{align*}
        &-\limsup_{T\rightarrow\infty} \frac{V_N^\ast(x_T, \hat{\theta}_T)-V_N^\ast(x_0, \hat{\theta}_0)}{T}+\limsup_{T\rightarrow\infty}\frac{\sum_{k=0}^{T-1} \ell_{0|k}^\ast}{T}\\
        &\quad\stackrel{\eqref{proof_asymp_intermediate-result}}{\leq}\limsup_{T\rightarrow\infty} \sum_{k=0}^{T-1} \frac{1}{T}\left[C_A \|w_k\|+C_A\|\tilde{x}_{1|k}\|\right] \stackrel{\eqref{assump_finite-energy_equation}, \eqref{lemma_asymp-lin-bound_equation}}{\leq} 0,
    \end{align*}
    which shows~\eqref{thm_asymp-perf_equation}.
    \end{proof}

The asymptotic performance bound ensures that on average, the closed-loop system will not perform worse than the operation at the origin, independent of the magnitude of the initial parameter error.

In the special case that the economic cost $\ell_\text{eco}(x,u)$ is positive definite, we recover an adaptive MPC with a tracking objective~\cite{lorenzenRobustMPCRecursive2019a}, and Theorem~\ref{thm_asymp-perf} recovers the stability result in \cite[Corollary~16]{lorenzenRobustMPCRecursive2019a}. 
\begin{corollary}[Asymptotic convergence]
    Let Assumptions~\ref{assump_system-matrices} and~\ref{assump_finite-energy} hold and suppose $\ell(x,u,s)$ is positive definite, i.e. $Q,R\succ 0$, $q=0,r=0$. Then, $\lim_{k\rightarrow\infty} \|x_k\| =0$ for the closed-loop system resulting from Algorithm~\ref{AEMPC-algorithm}.
\end{corollary}

\subsection{Transient performance bound}\label{sec_theory_trans}
The previously shown asymptotic average performance does not allow for a statement about the AE-MPC's transient performance, i.e., its performance over a finite time $T$. In this subsection, we study the finite-time behavior in more detail. 
\begin{assumption}[Point-wise bounded disturbances] \label{assump_dist-average-bound}
    There exists a constant $\bar{w}\geq0$ such that 
    \begin{align} \label{assump_dist-average-bound_equation}
        \sum_{k=0}^{T-1} \|w_k\|^2 \leq T\bar{w}^2 \quad \ \forall\,T \in\mathbb{N}.
    \end{align}
\end{assumption}
Under Assumption~\ref{assump_dist-average-bound}, the disturbances are bounded by a uniform constant at all times. This is weaker than assuming finite-energy, because we do not require $\|w_k\|^2\rightarrow0$ for $k\rightarrow\infty$.

\begin{theorem}[Transient performance bound] \label{thm_transient-perf}
     Let Assumptions~\ref{assump_system-matrices} and~\ref{assump_dist-average-bound} hold. Then, there exist uniform constants $C_V, C_A\geq 0$ such that for all $T\in\mathbb{N}$, the closed-loop system resulting from Algorithm~\ref{AEMPC-algorithm} satisfies
    \begin{align}\label{thm_transient-perf_equation}\begin{split}
        &\sum_{k=0}^{T-1} \ell(x_k,u_k, s_{0|k}^\ast) \leq C_V+2\bar{w}C_AT+\frac{C_A}{\sqrt{\mu}}\|\hat{\theta}_0-\theta^\ast\|\sqrt{T}.
    \end{split}
    \end{align}
\end{theorem}

\begin{proof}
    First, note that the derivation of \eqref{proof_asymp_intermediate-result} in the proof of Theorem~\ref{thm_asymp-perf} remains valid with the bound of Assumption~\ref{assump_dist-average-bound}.
    Using the LMS parameter estimate yields
    \begin{align} \label{proof_transient_LMS-bound}
        \sum_{k=0}^{T-1} \|\tilde{x}_{1|k}\|^2 \stackrel{\eqref{prop_LMS-bound_equation},\eqref{assump_dist-average-bound_equation}}{\leq} T \bar{w}^2+\frac{1}{\mu}\|\hat{\theta}_0-\theta^\ast\|^2.
    \end{align}
    The one-step prediction error $\tilde{x}_{1|k}$ satisfies
    \begin{align}\label{proof_transient_x-bound}
            \sum_{k=0}^{T-1} \|\tilde{x}_{1|k}\| 
            \leq \sqrt{T}\sqrt{\sum_{k=0}^{T-1}\|\tilde{x}_{1|k}\|^2}    \stackrel{\eqref{proof_transient_LMS-bound}}{\leq} T\bar{w}+\sqrt{\frac{T}{\mu}}\|\hat{\theta}_0-\theta^\ast\|,
        \end{align}
        where the first inequality used the fact that $\|X\|_1\leq \sqrt{T}\|X\|$ for any $X\in\mathbb{R}^T$ with $X_i=\|\tilde{x}_{1|i}\|$ and the second inequality makes use of ${\sqrt{a+b} \leq \sqrt{a}+\sqrt{b}}$ for any $a,b\geq0$. Similarly, for the disturbances $w_k$ we obtain
        \begin{align}
        \begin{split} \label{proof_transient_w-bound}
            \sum_{k=0}^{T-1} \|w_k\|
            &\leq \sqrt{T}\sqrt{\sum_{k=0}^{T-1} \|w_k\|^2}\stackrel{\eqref{assump_dist-average-bound_equation}}{\leq} \sqrt{T}\sqrt{T\bar{w}^2} = T\bar{w}.
        \end{split}
    \end{align}    
    Next, we combine these transient bounds with~\eqref{proof_asymp_intermediate-result} to obtain the transient performance bound:
    \begin{align*}
        \begin{split}
            &V_N^\ast(x_T, \hat{\theta}_T)-V_N^\ast(x_0, \hat{\theta}_0)\\
            &\qquad\stackrel{\eqref{proof_asymp_intermediate-result}}{\leq} \sum_{k=0}^{T-1}\left[-\ell_{0|k}^\ast + C_A \|\tilde{x}_{1|k}\| + C_A \|w_{k}\| \right]\\ 
            &\quad\, \stackrel{\eqref{proof_transient_x-bound},\eqref{proof_transient_w-bound}}{\leq} \sum_{k=0}^{T-1}\left[-\ell_{0|k}^\ast\right] +C_A \sqrt{\frac{T}{\mu}}\|\hat{\theta}_0-\theta^\ast\| + 2C_A \bar{w}T, 
        \end{split}
    \end{align*}
    where $C_A$ is defined after~\eqref{proof_asymp_intermediate-result}. Defining the uniform constant $-C_V\leq V_N^\ast(x_k, \hat{\theta}_k)-V_N^\ast(x_0, \hat{\theta}_0)~\forall k\in\mathbb{N}$, using compactness of $\Theta$ and $\mathbb{Z}$, yields~\eqref{thm_transient-perf_equation}.
\end{proof}

\subsection{Discussion} \label{sec_discuss}
The asymptotic performance bound of Theorem~\ref{thm_asymp-perf} shows that on average, our AE-MPC scheme performs no worse than operating the system at the origin. Our performance bound is relative to the performance at the origin due to centering the terminal cost around it (cf. Prop.~\ref{prop_term-cost}). Even with bounded energy disturbances and uncertain initial parameters $\hat{\theta}_0\neq\theta^\ast$, Theorem~\ref{thm_asymp-perf} establishes the same asymptotic performance guarantees as economic MPC without disturbances and with perfectly known parameters~\cite[Theorem~18]{amritEconomicOptimizationUsing2011}. The transient performance bound of Theorem~\ref{thm_transient-perf} provides insight into the behavior under persistent disturbances. This transient performance bound depends linearly on the magnitude of the disturbances $\bar{w}$. Furthermore, the system's transient performance depends on the initial parameter error. For ${\bar{w}\rightarrow 0}$ and ${T\rightarrow\infty}$, the transient performance bound (Thm.~\ref{thm_transient-perf}) recovers the asymptotic performance bound (Thm.~\ref{thm_asymp-perf}).

The AE-MPC scheme considers the special case of soft state constraints which is useful for a variety of systems (e.g., temperature control). Its implementation is neither conceptually nor computationally more complex than implementing nominal MPC with soft constraints.
We expect that the MPC scheme can be extended with a robust tube MPC formulation (e.g.,~\cite{kohlerLinearRobustAdaptive2020}) to ensure robust constraint satisfaction and relax open-loop stability to stabilizability of $(A(\theta),B(\theta))$.

\section{Numerical example} \label{sec_numerical-example}
We illustrate the theoretical results using a simple building temperature control example and study the performance benefits of online adaptation for economic MPC. The corresponding MATLAB code is available online.\footnote{\href{https://doi.org/10.3929/ethz-b-000690654}{https://doi.org/10.3929/ethz-b-000690654}}
\subsubsection*{Model and setup}
We model the thermal dynamics of a single temperature zone in a building using a 2R2C-model~\cite{wangDevelopmentRCModel2019}. 
The discrete-time model equations are given by
\begin{align*}
\begin{split}
    x_{k+1} \!=\! 
    \begin{bmatrix}
        1\shortminus\frac{[\theta]_1}{a_1} &\!\! \frac{[\theta]_1}{a_1}\\
        \frac{[\theta]_1}{a_2} &\!\! 1\shortminus\frac{[\theta]_1+[\theta]_2}{a_2} 
    \end{bmatrix}
    \!x_k \!+\!
    \begin{bmatrix}
        \frac{1}{a_1}\\
        0 
    \end{bmatrix}\! u_k
    \!+\!\begin{bmatrix}
        \frac{1}{a_1} &\!\! 0\\ 0&\!\! \frac{[\theta]_2}{a_2} 
    \end{bmatrix}\! w_k,
\end{split}
\end{align*}
where $a_1 = 0.6125, a_2=21.12$ are constants that depend on the mass and the specific heat capacity of the wall materials. 
The room and wall temperature are the system's states $[x]_1,\ [x]_2$, respectively.
The input $u$ is the controllable heat flow into the room and the disturbances $w$ consist of a disturbance heat flow $[w]_1$ (e.g., solar irradiation) and the outdoor temperature $[w]_2$. Normal operation, i.e., room temperature of $22^\circ\,$C and a nominal heat flow, corresponds to $x=0,\ u=0$. 
We use a sampling time of $10$ minutes. 
The parameters $\theta$ reflect uncertainty in the thermal resistance. The parameter set is given $\Theta=[0.005, 0.144]\times[0.0004, 0.0075]$, with $\theta^\ast=[0.048, 0.0015]^\top$, $\hat{\theta}_0=[0.1, 0.0075]^\top$.
The disturbances take values according to 
\begin{align*}
    [w_k]_1 &= 0.595\left(\sin\left(\frac{2\pi}{24\cdot6}k\right)+\eta_k(-0.2,0.2)\right),\\ 
    [w_k]_2 &= 7\left(\sin\left(\frac{2\pi}{24\cdot6}(k-9)\right)+\eta_k(-0.5,0.5)\right), 
\end{align*}
where $\eta_k(a,b)$ is the value of a random variable, uniformly distributed in the interval $[a,b]$, and the phase shift $k-9$ models a lag of $90$ minutes of the outdoor temperature relative to the solar irradiation.
The period of the sinus is one day. To simplify the exposition and in accordance to the setting in this paper, the disturbances and their prediction are completely unknown to the controller.

We consider a heating problem and the room temperature is supposed to satisfy $[x]_1\geq -1$. The input is constrained by $\mathbb{U}=[-0.8, 0.8]$ and the economic cost of the system is $\ell_\text{eco}(x,u) = x^\top Qx + r^\top u$, where $Q=\diag(500,0), r = 600$. 
We want to keep the room temperature in a comfortable range and minimize the energy consumption; the factor $r>0$ reflects that increasing heating $u$ increases the cost.
The soft penalty on the temperature constraint is given by $\lambda=10^5$ and we use a horizon of $N=18$. 
The parameter update gain is $\mu=5\cdot 10^{-5}$, which satisfies~\eqref{LMS_gain-definition} with the conservative bound $x_k\in\mathbb{Z} = \{ [x]_1\in [-20,20], [x]_2 \in [-10,10]\}$. 
Note that this problem satisfies the conditions in this paper, in particular, Assumptions~\ref{assump_system-matrices} and \ref{assump_dist-average-bound} hold.

\subsubsection*{Results and comparison}
In the following, we compare the performance of the proposed AE-MPC scheme (Alg.~\ref{AEMPC-algorithm}) to the performance of economic MPC (E-MPC) without parameter adaptation, i.e., $\hat{\theta}_k=\hat{\theta}_0, \forall k\geq 0$.

The evolution of the parameter estimate $\hat{\theta}_k$ is shown in Figure~\ref{fig_LMS}. 
The parameter estimate improves over a few of days and the error to the true parameter reduces. After 20 days, the parameter estimate $\hat{\theta}_k$ does not change significantly anymore. 
However, since the considered disturbances are correlated and no active excitation is performed, some residual bias remains. 

In Figure~\ref{fig_state-input}, we show the evolution of the room temperature $[x]_1$ and the control input $u$ for two three-day periods (days zero to three and 20 to 23). Figure~\ref{fig_perf} depicts the accumulated costs. 
During an initial phase, the closed-loop evolution of AE-MPC and E-MPC is almost indistinguishable but afterwards, AE-MPC achieves significantly smaller costs, fewer constraint violations, and smaller oscillations of the room temperature due to the improved model. 
These results are consistent with Theorem~\ref{thm_transient-perf}, showing that the effect of the initial parameter error vanishes as $T\rightarrow\infty$.

Overall, this numerical example illustrates that the proposed AE-MPC enhances the economic performance during operation under persistent disturbances and parametric uncertainty.

\begin{figure}[!t]
    \centering
    \includegraphics[width=0.45\textwidth]{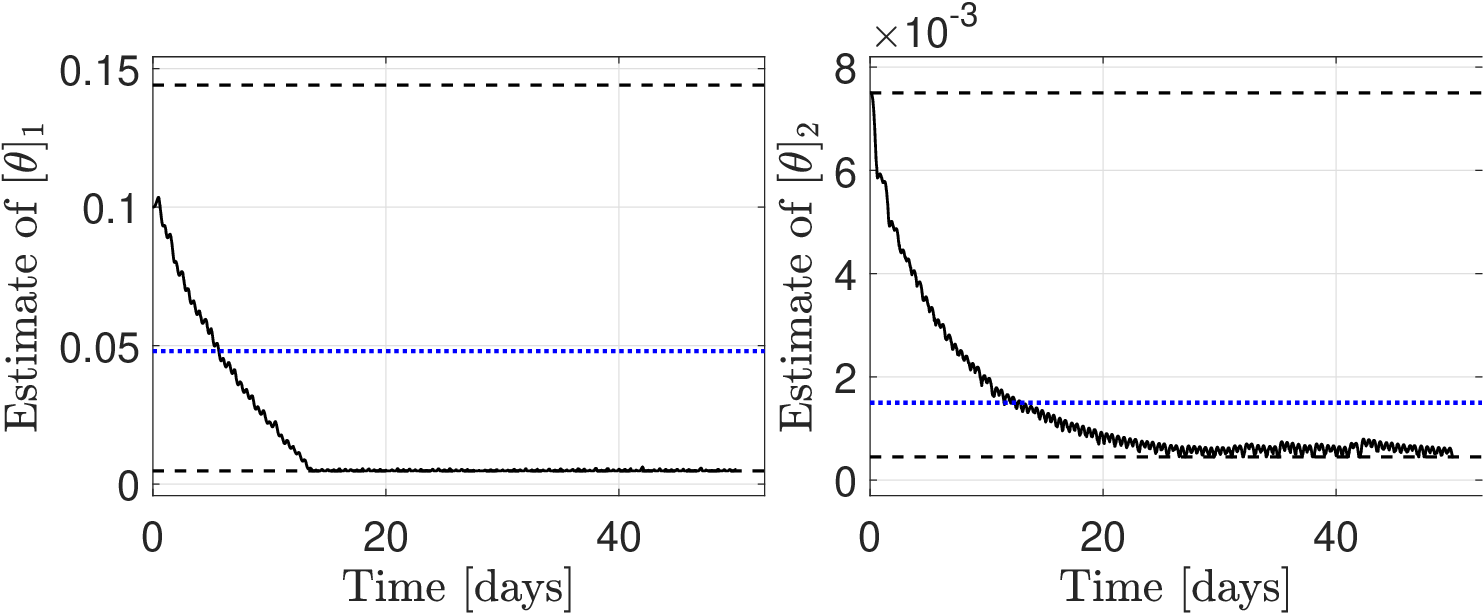}
    \caption{Evolution of the parameter estimate (black solid). Bounds of the parameter set are the dashed black lines, and the dotted blue line is the true parameter value.}
    \label{fig_LMS}
\end{figure}

\begin{figure}[!t]
    \centering 
    \includegraphics[width=0.48\textwidth]{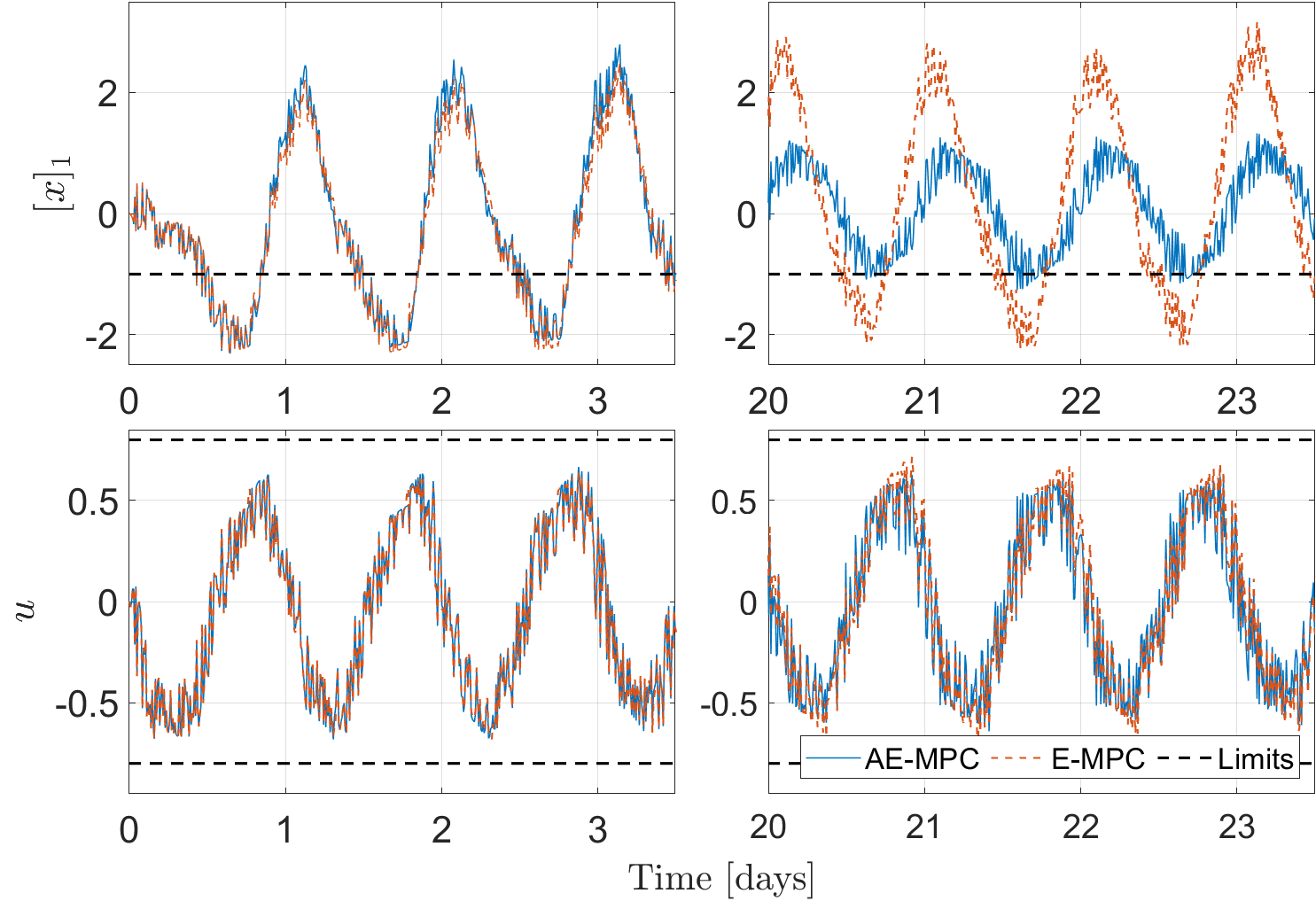}
    \caption{Evolution of the state $[x]_1$ (top) and the input $u$ (bottom) for the initial phase (left) and after parameter convergence (right). The constraint sets $\mathbb{X}, \mathbb{U}$ are shown as black, dashed lines.}
    \label{fig_state-input}
\end{figure}

\begin{figure}[!t]
    \centering
    \includegraphics[width=0.45\textwidth]{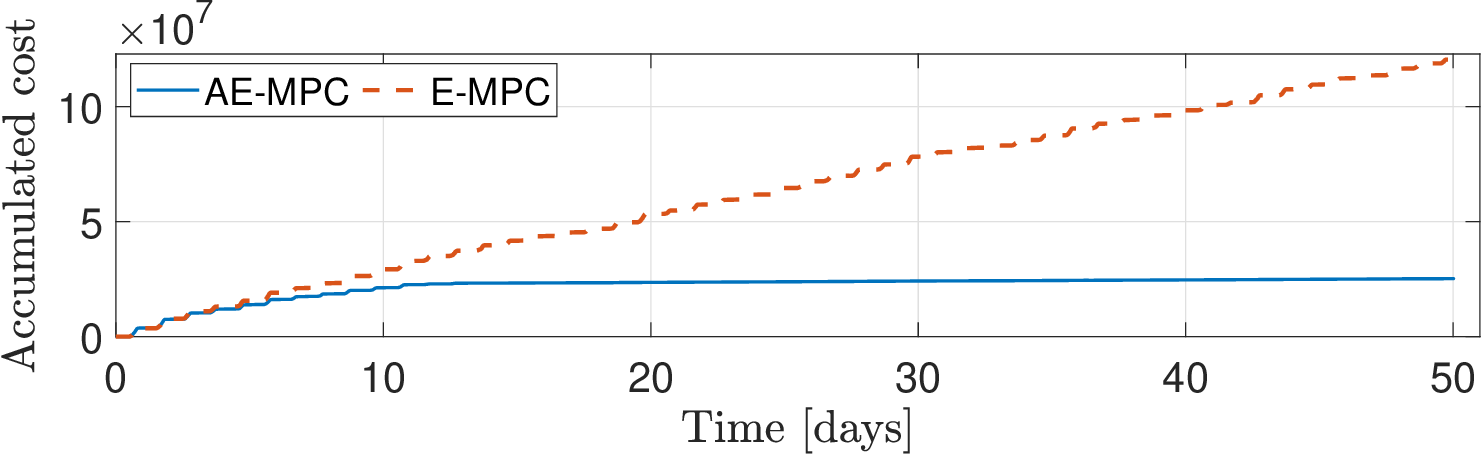}
    \caption{Accumulated cost of AE-MPC (blue) and E-MPC (orange).}
    \label{fig_perf}
\end{figure}

\section{Conclusion} \label{sec_conclusion}
A computationally and conceptually simple adaptive economic MPC scheme has been proposed for linear systems with parametric uncertainty and additive bounded disturbances. We used a projected LMS filter for the online model adaptation and derived theoretical performance guarantees for the asymptotic and transient cases. In particular, we showed for finite-energy disturbances that, asymptotically, AE-MPC has the same performance guarantee as nominal economic MPC with perfect knowledge of the parameters. 

Future work includes parameter dependent optimal steady-states using ideas from adaptive MPC~\cite{sasfi2023robust}, more general economic MPC schemes~\cite[Chap. 5]{kohler2024analysis}, and the treatment of (slowly) time-varying true parameters.

\bibliographystyle{IEEEtran}  
\bibliography{literature}  

\section*{Appendix}
\begin{lemma}[Limits of infinite series~{\cite[Lemma~1]{solopertoGuaranteedClosedLoopLearning2023}}]\label{lemma_infinite-series}
    For any sequence $(a_k)$ that satisfies $0\leq a_k\leq a_{\max} < \infty\ \forall k\in\mathbb{N}$ and any function $\alpha \in \mathcal{K}_\infty$, the following implication holds: 
    \begin{align} \label{lemma_infinite-series_equation_appendix}
        \lim_{T\rightarrow\infty}\sum_{k=0}^{T-1} a_k \leq S < \infty \Rightarrow \lim_{T\rightarrow\infty}\sum_{k=0}^{T-1} \frac{\alpha(a_k)}{T} = 0
    \end{align}
\end{lemma}
\begin{lemma}[Asymptotic convergence of prediction error]\label{lemma_asymp-lin-bound}
    Let Assumption~\ref{assump_finite-energy} hold. Then, 
    \begin{align*}
         \lim_{T \rightarrow \infty} \frac{\sum_{k=0}^{T-1} \|\tilde{x}_{1|k}\| + \|w_k\|}{T} =0.
    \end{align*}    
\end{lemma}
\begin{proof}
    We know by assumption that $\sum_{k=0}^T \|w_k\|^2 \leq S_{\mathrm{w}} < \infty$. 
    We can now use Lemma~\ref{lemma_infinite-series} 
    with $a_k = \|w_k\|^2$ and $\alpha(s) = \sqrt{s}, \alpha\in \mathcal{K}_\infty$ to get:
\begin{align}
        \lim_{T\rightarrow\infty}\sum_{k=0}^{T-1} \frac{\sqrt{\|w_k\|^2}}{T} = \lim_{T\rightarrow\infty}\sum_{k=0}^{T-1} \frac{\|w_k\|}{T}= 0  \label{lemma_bound_proof_a}.
    \end{align}
    With $\Theta$ compact, we have $\frac{1}{\mu}\|\hat{\theta}_0-\theta^\ast\|^2<\infty$, implying 
    \begin{align*}
       \sum_{k=0}^{T-1} \|\tilde{x}_{1|k}\|^2 \stackrel{\eqref{prop_LMS-bound_equation}}{\leq} S_{\mathrm{w}} + \frac{1}{\mu}\|\hat{\theta}_0-\theta^\ast\|^2 < \infty .
        \end{align*}
    Applying  Lemma~\ref{lemma_infinite-series}, 
    with $a_k = \|\tilde{x}_{1|k}\|^2$ and $\alpha\in\mathcal{K}_\infty$ as above, yields   
    \begin{align}
         \lim_{T\rightarrow\infty}\sum_{k=0}^{T-1} \frac{\sqrt{\|\tilde{x}_{1|k}\|^2}}{T}
        =\! \lim_{T\rightarrow\infty}\sum_{k=0}^{T-1} \frac{\|\tilde{x}_{1|k}\|}{T} = 0  .\label{lemma_bound_proof_b}
    \end{align}
    Adding~\eqref{lemma_bound_proof_a} and~\eqref{lemma_bound_proof_b} yields the desired inequality.
      \end{proof}

\end{document}